\documentclass[a4paper,12pt,leqno]{article}

\input{style}
\usepackage{todonotes}

\begin{document}

\title{Superreplication when trading at market indifference prices.}

\author{P. Bank and S. G{\"o}kay\\  Technische Universit{\"a}t Berlin\\
  Institut f{\"u}r Mathematik\\
  Stra{\ss}e des 17. Juni 135, 10623 Berlin, Germany \\
  (bank@math.tu-berlin.de, gokay@math.tu-berlin.de) \thanks{MATHEON, Einstein}} 

\date{\today}

\maketitle

\begin{abstract}
  We study superreplication of European contingent claims in discrete
  time in a large trader model with market indifference prices
  recently proposed by Bank and Kramkov.  We introduce a suitable
  notion of efficient friction in this framework, adopting a
  terminology introduced by Kabanov, Rasonyi, and Stricker in the
  context of models with proportional transaction costs. In our
  framework, efficient friction ensures that large positions of the
  investor may lead to large losses, a fact from which we derive the
  existence of superreplicating strategies. We illustrate that without
  this condition there may be no superreplicating strategy with
  minimal costs. In our main result, we establish efficient friction
  under a tail condition on the conditional distributions of the
  traded securities and under an asymptotic criterion on risk
  aversions of the market makers. Another result asserts that strict
  monotonicity of the conditional essential infima and suprema of the
  security prices is sufficient for efficient friction.  We give
  examples that satisfy the assumptions in our conditions, which
  include non-degenerate finite sample space models as well as Levy
  processes and an affine stochastic volatility model of
  Barndorff-Nielsen-Shepard type.
\end{abstract}

\begin{description}
\item[Keywords:] utility indifference
  prices, large investor, liquidity, superreplication, monotone
  exponential tails
\item[JEL Classification:] G11, G12, G13, C61. 
\item[AMS Subject Classification (2010):] 52A41, 60G35, 90C30, 91G20, 97M30.
\end{description}

\section{Introduction}

The problem of superreplicating a contingent claim has been widely
studied in Mathematical Finance. For frictionless diffusion models
this stochastic control problem was first addressed by
\citet{eq1995} and, for general semimartingales, by
\citet{k1996}. Also models with market frictions have received a lot
of attention. For markets with portfolio constraints such as the
prohibition of short selling see, e.g., \citet{ck1993}, \citet{jk1995},
\citet{fk1997}, \citet{bcs1998}. For markets with proportional transaction
costs similarly far reaching investigations have resulted from the
work of, e.g., \citet{ssc1995, cpt1999, ks2002, cs2006,  grs2008}. For a
complete treatment and list of references, we refer to the book, 
\citet{ks2009}. Recently, superreplication has also been studied in
nonlinear models capturing illiquidity effects by, e.g., \citet{cst2010, gs2012, 
pst2012, ds2013}.

The present paper focusses on the nonlinear large investor model with
market indifference prices developed in \citet{bk2011a,
  bk2011b, bk2013}. By contrast to the model discussed in
\citet{cjp2004, cst2010, gs2012, ds2013}, this model
does not postulate a local cost term depending on the size of the
current transaction which would be attributed to a \emph{temporary}
market impact. Instead, market indifference prices can be viewed as a
way to specify systematically the \emph{permanent} price impact of a
transaction. While the impact is adverse in the sense that a large buy
order will substantially drive up marginal prices, it is far from
obvious when this impact actually results in an \emph{efficient
  friction}, i.e., in real costs to the large investor. Indeed, if
after a purchase no new information becomes available about the
ultimate value of the traded securities, a subsequent sale of the same
position will take place at the same marginal prices, now processed in
reverse order. Hence, the sale will recover all the expenses incurred
from the initial purchase leading to a `free roundtrip'. This is
similar to a phenomenon already observed in multi-variate asset price
models with proportional transaction costs where a suitable notion of
efficient friction was introduced by \citet{krs2002} to ensure
that trades actually incur costs. 

With the same purpose in mind, we adopt this terminology for our
framework despite the mentioned differences in the models and we thus
say that, essentially, a model with market indifference prices
exhibits efficient friction if an investor engaging in ever larger
positions may face ever higher losses with positive probability. A
simple two-period example with a single market maker shows that
without this weak condition superreplicating strategies may fail to
exist. Our main result, Theorem~\ref{thm:1}, thus develops a readily
verifiable criterion for efficient friction to hold. This criterion is
based on tail conditions for both the conditional distributions of
terminal values of traded securities and the asymptotic risk aversions
of market makers. Corollary~\ref{cor:1} then shows that efficient
friction indeed implies the existence of superreplicating
strategies. In models where conditional essential infima and suprema
of security prices are attained on sets with positive probability,
Theorem~\ref{thm:2} shows that the strict monotonicity of these
extrema is also sufficient for efficient friction, even without extra
assumptions on the market makers' asymptotic risk aversions. We also
characterize when a binomial model with an exponential market maker is
complete and find again that this depends on a monotonicity condition
on conditional extrema. As examples where our tail condition on
conditional distributions holds we consider L{\'e}vy processes and a
Barndorff-Nielsen-Shepard model with stochastic volatility; see
Section~\ref{sec:examples}.

\section{Problem formulation and motivation}

Before we can properly formulate the superreplication problem we want to
address in Section~\ref{sec:problem}, we first have to introduce
the modeling framework we are going to use.

\subsection{Trading at market indifference prices}\label{sec:model}

For our model we shall use a discrete-time version of the framework
introduced in~\citet{bk2011a, bk2011b} which we shall outline
briefly in this section for the reader's convenience. Specifically, we
consider a financial market where $M \in \cbr{1,2,\ldots}$ market makers
quote prices for $J \in \cbr{1,2,\dots}$ securities with respective random payoffs
$\psi=(\psi^1,\ldots,\psi^J)$ at time $T$. For simplicity, the market
makers have a common view of uncertainty which is described by a
filtered probability space $(\Omega,\cF,(\cF_t)_{t=0,\ldots,T},\P)$
where $\cF_0 = \cbr{\varnothing,\Omega}$ up to $\P$-null sets; in
particular $\psi \in \mathbf{L}^0(\cF_T, \mathbb{R^J})$.
The market makers, however, may have different attitudes towards risk:

\begin{Assumption}\label{asp:1}
Each market maker $m=1,\ldots,M$ has a utility function $u_m :
\mathbb{R} \to \mathbb{R}$ which is strictly concave, increasing,
twice continuous differentiable with
$$
\lim_{x \uparrow \infty} u_m(x)=0.
$$
Moreover, absolute risk aversion is bounded in the sense that
\begin{equation}\label{eq:1}
\frac1c \leq a_m(x) \set -\frac{u_m''(x)}{u'_m(x)} \leq c, \; x \in \mathbb{R}, \text{ for
some } c>0.
\end{equation}
\end{Assumption}

The market makers shall be allowed to trade freely among themselves
(e.g. using a complete OTC market). As a result, they always allocate
their total wealth in a (conditionally) Pareto optimal way, i.e., such
that no reallocation of wealth is possible which would leave one
market maker better off and none of them worse off in terms of (conditional) 
expected utilities. It is well-known that mathematically such allocations 
are most conveniently described as the
maximizers of the representative agent's utility function
\begin{equation}\label{eq:2}
r(v, x) \set \sup_{x^1 + \ldots + x^m = x} \sum_{m =1}^M v_m u_m(x^m), \quad x \in \mathbb{R},
\end{equation}
where $v \in (0,\infty)^M$ assigns weights to our market makers. For
instance, if $\alpha_0 = (\alpha^m_0)_{m=1,\ldots,M}$ denotes the
Pareto optimal initial allocation of wealth among the market makers,
there is a vector $v_0 \in (0,\infty)^M$, unique up to scaling by a
positive constant, such that
$$
v^m_0 u_m'(\alpha^m_0) = \partial_x r(v_0,\Sigma_0), \; m=1,\ldots,M,
$$
where 
$$
\Sigma_0 \set \sum_{m}^M \alpha^m_0
$$ 
denotes the market makers' total initial endowment; see, e.g.,
Lemma~3.2 in~\cite{bk2011a}.

More generally, if by time $t=0,\ldots,T$ the market makers have
jointly acquired in addition $x \in \RR$ units of cash and $q \in \RR^J$
securities $\psi$, the resulting total endowment
$$
\Sigma(x,q) \set \Sigma_0 + x + \abr{q,\psi}
$$ 
will lead to the representative agent's expected utility
$$
F_t(v,x,q) \set \condexp{r(v,\Sigma(x,q))}.
$$
Theorem~4.1 in \cite{bk2013} shows that for each $t = 0, \ldots, T$, this is a saddle
function of class $C^2$ in the variables
$$
(v,x,q) \in \mathbf{A} \set (0,\infty)^M \times \RR \times \RR^J,
$$ 
if, in addition to Assumption~\ref{asp:1}, we impose the following
assumption:
\begin{Assumption}\label{asp:2}
  For any $x \in \mathbb{R}$ and $q \in \mathbb{R}^J$ there is an
  allocation $\beta \in \mathbf{L}^0(\mathcal{F}_T, \mathbb{R}^M)$ with total
  endowment $\Sigma(x,q)$ such that
\begin{equation*}
\E \left[ u_m(\beta^m) \right] > - \infty, \ \ m=1, 2, \ldots, M.
\end{equation*}
\end{Assumption}

Theorem~4.1 in~\cite{bk2013} shows furthermore that also 
$$
G_t(u,y,q) \set \sup_{v \in (0,\infty)^M} \inf_{x \in \RR}
\cbr{\abr{u,v}+xy-F_t(v,x,q)}
$$
is a twice continuously differentiable saddle function of the
variables
$$
(u,v,q) \in \mathbf{B} \set (-\infty,0)^M \times (0,\infty) \times \RR^J
$$
and that $G_t$ and $F_t$ are conjugate in the sense that
conversely
$$
F_t(v,x,q) = \sup_{v \in (0,\infty)^M} \inf_{x \in \RR}
\cbr{\abr{u,v}+xy-G_t(u,y,q)}, \; (v,x,q) \in \mathbf{A}.
$$

In the sequel, we shall study how a single large investor can trade
with the market makers in order to hedge against a liability $H$ with
maturity $T$ by following a judiciously chosen dynamic strategy. So,
let $Q=(Q_t)_{t=1,\ldots,T}$ denote the predictable positions the
large investor will ask our market makers to hold in the marketed
securities $\psi$. We then have to describe the predictable cash
balance $X=(X_t)_{t=1,\ldots,T}$ the market makers will ask for as
compensation. As pointed out in~\citet{bk2011a, bk2011b} one
natural possibility is to have $X_t \in \mathbf{L}^0(\cF_{t-1},\RR)$
determined by utility indifference, i.e., by requiring that
\begin{equation}\label{eq:3}
  U^m_{t-1} = \condexp[\cF_{t-1}]{u_m(\alpha^m_t)}, \; m=1,\ldots,M,
\end{equation}
where $U_{t-1}= (U^m_{t-1})_{m=1,\ldots,M} \in
\mathbf{L}^0(\cF_{t-1},(-\infty,0)^M)$ records the market makers'
conditional expected utilities before the transaction and where
$\alpha_t = (\alpha^m_t)_{m=1,\ldots,M} \in
\mathbf{L}^0(\cF_{T},\RR^M)$ denotes the (as it turns out) only Pareto
optimal allocation of
$$
\Sigma_t \set \Sigma(X_t,Q_t) = \Sigma_0 + X_t+\abr{Q_t,\psi}
$$
for which the indifference relation~\eqref{eq:3} holds. In fact, as
shown by Theorem~4.1 in~\cite{bk2011b}
this cash balance is given by
\begin{equation}\label{eq:4}
  X_t = G_{t-1}(U_{t-1},1,Q_t)
\end{equation}
 and the Pareto allocation is determined by the weights
\begin{equation}\label{eq:5}
  V_t = \partial_u G_{t-1}(U_{t-1},1,Q_t)\,.
\end{equation}
Moreover, the market makers' utilities at time $t=1,\ldots,T$ are
given by
\begin{equation}\label{eq:6}
  U_t = \partial_v F_t(V_t, X_t, Q_t) = (\condexp{u_m(\alpha^m_t)})_{m=1,\ldots,M}\,.
\end{equation}
Hence, for any strategy $Q$ the dynamics of the system are uniquely
determined by the initial level of our market makers' utilities and
equations~\eqref{eq:4}, \eqref{eq:5}, and~\eqref{eq:6}.

It will turn out to be convenient to denote by
$U^{s,u,Q}=(U^{s,u,Q}_t)_{t=s,\ldots,T}$,
$X^{s,u,Q}=(X^{s,u,Q}_t)_{t=s+1,\ldots,T}$, and
$V^{s,u,Q}=(V^{s,u,Q}_t)_{t=s+1,\ldots,T}$ the evolution of the system
with~\eqref{eq:4}, \eqref{eq:5}, and~\eqref{eq:6} when started at
$$
U^{s,u,Q}_s \set u \in \mathbf{L}^0(\cF_s,(-\infty,0)^M)
$$
at some time $s \in \cbr{0,\ldots,T}$.

\subsection{The superreplication problem of a large investor}\label{sec:problem}

With the market dynamics defined for any predictable strategy, we are
now in a position to formulate the large investor's superreplication
problem.  As usual we shall say that an initial capital $\pi$ suffices
to superreplicate a contingent claim with payoff $H \in
L^0(\cF_T,\RR)$ at time $T$ if there is a strategy $Q$ which generates
profits or losses $PL^Q_T$ by time $T$ such that
\begin{equation}\label{eq:7}
  H \leq \pi + PL^Q_T\,.
\end{equation}
By construction of our market model, the large investor's gains are
the market makers' joint losses and so~\eqref{eq:7} can be recast as
the requirement that
\begin{equation}\label{eq:8}
  X^{0,u_0,Q}_T+\abr{Q_T,\psi} \leq \pi - H\,
\end{equation}
where $u_0 = (\E u_m(\alpha^m_0))_{m=1,\ldots,M}$ denotes the initial
expected utility levels for our market makers.  So the
superreplication price of $H$ turns out to be
\begin{equation}\label{eq:9}
  \pi^H \set \inf\cbr{ \pi \in \RR \;:\; \eqref{eq:8} \text{ holds for
  some predictable } Q=(Q_t)_{t=1,\ldots,T}}\,.
\end{equation}
\begin{Remark} \label{rem:1}
As usual, relations such as~\eqref{eq:7} and~\eqref{eq:8} are tacitly
understood in the $\P$-almost sure sense.
\end{Remark}

Obviously, if for some initial capital $\pi$ a strategy $Q$ can be
found which replicates $H$, i.e., for which we obtain equality
in~\eqref{eq:7} or, equivalently, \eqref{eq:8}, we have $\pi^H \leq
\pi$. In fact, we would then have
equality in this relation in arbitrage free, linear, discrete-time models typically
considered in Mathematical Finance. It may thus be interesting to note
that in the highly nonlinear model under investigation here this may
very well not be the case. Indeed, as pointed out in Remark~3.12 of
\citet{bk2011b}, there can be two strategies $Q$, $Q'$ which,
starting from different initial capitals $x<x'$, give the large
investor the same terminal wealth $x+PL^{Q}_T = x'+PL^{Q'}_T\set H$
without this violating absence of arbitrage. As a result, in our
nonlinear illiquid financial model pricing $H$ by replication may not
be the appropriate concept to investigate. By contrast, the notion of
superreplication clearly still makes sense in such a setting but it
becomes an issue whether one can ensure existence of a strategy which
superreplicates at the superreplication price $\pi^H$, i.e., whether
the infimum in~\eqref{eq:9} is actually a minimum. This issue will be
addressed by our main results.

\section{Main results}\label{sec:main-results}

The main goal of this paper is to identify readily verifiable
conditions on the traded payoffs and the market makers' risk
preferences ensuring that any claim $H$ can be superreplicated
starting from the superreplication price. This will be accomplished in
Theorem~\ref{thm:1} and its Corollary~\ref{cor:1}
below. Theorem~\ref{thm:1} identifies conditions on the tails of both
the payoff's conditional distribution and of the market makers'
utilities that ensure a form of efficient friction to hold in our
model. Corollary~\ref{cor:1} then shows that in models with efficient
friction optimal superreplication strategies exist. 

\subsection{Exponential tails decreasing in time}

Our first condition ensures that the market makers' assessment of the
riskiness of the payoff $\psi$ may change sufficiently between any two
trading periods. The counterexample given in
Section~\ref{sec:counterexample} shows that a condition of this nature
is in fact necessary even when we restrict ourselves to the
particularly simple case of a single market maker with exponential
utility.

To formulate our condition we introduce the following (partial)
ordering relation between any two distributions $\mu$, $\nu$ on
$(\RR^J,\cB(\RR^J))$ with finite exponential moments:
\begin{equation}\label{eq:10}
\mu \prec \nu \quad :\Longleftrightarrow \quad \lim_{|q| \uparrow
  \infty} \frac{\int \exp{\left(\abr{q,x}\right)} \mu(dx)}{\int \exp{\left(\abr{q,x}\right)} \nu(dx)}=0\,.
\end{equation}
Thus $\mu \prec \nu$ will hold if the exponential tails of~$\mu$ are
dominated by those of~$\nu$. 

The condition we shall impose on the payoff profile $\psi$ amounts to
the requirement that the conditional distributions of $\psi$ along the
filtration $(\cF_t)_{t=0,\ldots,T}$,
\begin{align}\label{eq:11}
\nu_t \set \condprob{\psi \in \cdot}, \; t=0,\ldots,T.
\end{align}
have the potential to decrease at any time in the sense that
\begin{equation}\label{eq:12}
\mathbb{P}\left[\nu_t \prec \nu_{t-1} | \mathcal{F}_{t-1} \right] > 0 \text{ for all } t=1,\ldots,T\,.
\end{equation}

In Section~\ref{sec:examples} below we shall verify~\eqref{eq:12} if
$\psi$ is the value at time $T$ of a Brownian motion or even of a
L{\'e}vy process (monitored at discrete points in time).  Similarly,
one can actually consider terminal values of affine stock price models
such as the Barndorff-Nielsen-Shepard style model presented in the
same section.

\subsection{Asymptotic risk aversion}

The second condition we have to impose focusses on the market makers'
preferences. It essentially amounts to the requirement that their
absolute risk aversion at $-\infty$ stabilizes at a higher level than
at $+\infty$. For a generic strictly concave, increasing utility
function $u \in C^2(\RR)$ this amounts to the condition that the
absolute risk aversion $a(x)\set-{u''(x)}/{u'(x)}$ satisfies
\begin{equation}\label{eq:13}
\int_{-\infty}^0 |\overline{a}-a(x)|
\,dx+\int_0^\infty |\underline{a}-a(x)| \,dx<\infty \text{ for some }
0<\underline{a} \leq \overline{a}<\infty\,.
\end{equation}
It is easy to see that utility functions which are mixtures of
exponential utilities of the form
$$
u(x) \set -\int_{\underline{a}}^{\overline{a}} \exp(-ax)
\,\Upsilon(da), \; x \in \RR,
$$
satisfy condition~\eqref{eq:13}, e.g., if the finite Borel measure
$\Upsilon$ charges both $\underline{a}$ and $\overline{a}$. Note that
the $\sup$-convolution describing the representative agent's
utility~\eqref{eq:2} will inherit this property when each market
maker's utility satisfies it; see Lemma~\ref{lem:2} below.

\subsection{Efficient friction and existence of superreplication
  strategies}
  
In their investigation of the superreplication problem under
proportional transaction costs, \citet{krs2002}
introduced a notion of efficient friction which ensured, essentially,
that trading incurs costs. We adopt this idea and terminology for the
purposes of our nonlinear model in the following definition.

\begin{Definition} \label{def:1}
  A financial model in the framework \cite{bk2011a} exhibits efficient friction, if for any time
  $t=1,\ldots,T$, for any choice of utility levels $u^n \in \cF_{t-1}$
  with $-\infty < \inf_{m,n} u^n_m \leq \sup_{m,n} u^n_m<0$, and for
  any sequence of strategies $Q^n$ such that $\cbr{\lim_n
  |Q^n_t|=+\infty}$ has positive probability, also the large
  investor's losses $X^{t-1,u^n,Q^n}_T+\abr{Q^n_T,\psi}$ converge
  to $+\infty$ in probability on a set with positive probability.
\end{Definition}

With this notion at hand, we are now in a position to state the main result of this paper which
we will prove in Section~\ref{sec:mainproof}:

\begin{Theorem}\label{thm:1}
  Let Assumptions~\ref{asp:1} and~\ref{asp:2} hold and assume that the
  market makers' total initial endowment is of the form $\Sigma_0 =
  \widetilde{\Sigma}_0+\abr{q_0,\psi}$ for some bounded random
  variable $\widetilde{\Sigma}_0$. Then our model exhibits efficient
  friction, if $\psi$ exhibits potentially decreasing exponential
  tails in the sense that condition~\eqref{eq:12} holds and if, in
  addition, our market makers' risk aversions stabilize at a higher
  level at $-\infty$ than at $+\infty$ in the sense that they
  satisfy~\eqref{eq:13} for constants $0<\underline{a}_m \leq
  \overline{a}^m<\infty$, $m=1,\ldots,M$.
\end{Theorem}

As a corollary, let us note that indeed efficient friction ensures the existence of optimal superreplicating
strategies.

\begin{Corollary}\label{cor:1}
  Let Assumptions~\ref{asp:1} and~\ref{asp:2} hold. If our model exhibits efficient friction, 
  any contingent claim $H$ can be superreplicated
  by an investment strategy $Q^H$ with minimal initial capital $\pi=\pi^H$
  as in~\eqref{eq:9}.
\end{Corollary}
\begin{proof}
Recalling the convention that $\inf \varnothing = +\infty$, we can
assume that there are finite $\pi^n$ and strategies $Q^n$ such that
\begin{equation}\label{eq:14}
\pi^n \downarrow \pi^H \ttext{ and } X^{0,u_0,Q^n}_T + \abr{Q^n_T,\psi}
\leq \pi^n-H, \text{a.s.} \; n=1,2,\ldots \,.
\end{equation}
In particular, we have that the large investor's losses are bounded
from above uniformly in $n$: 
\begin{equation}\label{eq:15}
  \sup_{n=1,2,\ldots} \cbr{X^{0,u_0,Q^n}_T +
    \abr{Q^n_T,\psi}}<\infty \text{ a.s}\,.
\end{equation}

We shall proceed inductively to construct a limiting strategy $Q^H$
which superreplicates $H$ starting with the minimal initial capital
$\pi^H$. In fact, for $t=1$ we can apply the efficient friction
property (with $u^n \set u_0$, $n=1,2,\ldots$) to obtain
that~\eqref{eq:15} rules out that $|Q^{n'}_1|$ explodes to $+\infty$
with positive probability along any subsequence $n'$. Hence, $\sup_n
|Q^n_1|<\infty$ and so we can choose a subsequence, again denoted by
$Q^n$, such that $Q^*_1 \set \lim_n Q^n_1$ exists in
$\mathbf{L}^0(\cF_0,\RR^J)$ (see, e.g., Lemma~1.64 in \citet{fs2011}).  The
continuity of our system dynamics specifying cash balances
$X^{0,u,Q^n}$ (cf.~\eqref{eq:4}), weights $V^{0,u,Q^n}$
(cf.~\eqref{eq:5}), and utilities $U^{0,u,Q^n}$ (cf.~\eqref{eq:6})
then ensures that the induced utilities $U^{0,u,Q^n}_1$ are bounded
away from zero and $-\infty$. We thus can now let $u^n \set
U^{0,u,Q^n}_1$, $n=1,2,\ldots$, and, observing that then
$X^{0,u_0,Q^n}_T=X^{1,u^n,Q^n}_T$, proceed successively in just the
same way to construct a limiting $Q^H_t \in \mathbf{L}^0(\cF_{t-1},\RR^J)$ for
$t=2$. The same reasoning applies for $t=3,\ldots,T$.

Since along the way the superreplication property~\eqref{eq:14} is
preserved for all the successive subsequences, the limiting strategy
$Q^H$ will also superreplicate but only need the minimal initial
capital $\pi^H$. The proof is accomplished.
\end{proof}

\subsection{Proof of Theorem~\ref{thm:1}}\label{sec:mainproof}

The proof of Theorem~\ref{thm:1} relies on three auxiliary
results. The first shows that when the market makers' risk aversion
stabilize asymptotically the same is true for the representative
agent's risk aversion:

\begin{Lemma}\label{lem:1}
  Assume the market makers have utility functions $u_m$,
  $m=1,\ldots,M$ satisfying Assumption~\ref{asp:1} and suppose in
  addition that risk aversions $a_m(x) \set -u''_m(x)/u'_m(x)$
  stabilize at $\pm \infty$ in the sense that~\eqref{eq:13} holds for
  $u \set u_m$ with respective constants $\underline{a} =\underline{a}_m$,
  $\overline{a}=\overline{a}_m$, $m=1,\ldots,M$.

  Then for any choice of $v \in (0,\infty)^M$ also the representative
  agent's utility $r(v,.)$ of~\eqref{eq:2} exhibits stabilizing
  risk aversion in the same sense namely at levels $\underline{a}$,
  $\overline{a}$ given by
  \[
  \frac1{\underline{a}} = \sum_{m=1}^M \frac1{\underline{a}_m} \text{
    and } \frac1{\overline{a}} = \sum_{m=1}^M
  \frac1{\overline{a}_m}\,.
  \]
\end{Lemma}
\begin{proof}
 Let $a_r(v,x) \set -\partial_x^2 r(v,x)/\partial_x r(v,x)$, $x \in \RR$, denote the representative agent's 
absolute risk aversion. It is straightforward to check (see, e.g.,
\cite{bk2011a}) that
\begin{equation*} 
a_r(v,x) = \frac{1}{\sum_{m=1}^M
\frac{1}{a_m(\widehat{x}^m(v,x))}}\,,
\end{equation*} 
where $(\widehat{x}^m(v,x))_{m=1,\ldots,M}$ denotes the unique point
in $\RR^M$ at which the $\sup$ in~\eqref{eq:2} is attained.

Observing that the function $f(a_1, a_2, \ldots, a_m) \set \frac{1}{\sum_{m=1}^M
\frac{1}{a_m}}$ is Lipschitz continuous in $(a_m)_{m=1,\ldots,M} \in (0,\infty)^M$ with
constant $1$ for the $\|.\|_1$-norm on $\RR^M$ we first note that
\[
|a_r(v,x)-\underline{a}| = \left|f(a_1(\widehat{x}^1),
\ldots, a_M(\widehat{x}^M)) - f(\underline{a}_1, \ldots, \underline{a}_M)\right| \leq
\sum_{m=1}^M \left|a_m(\widehat{x}^m) - \underline{a}_m\right|\,.
\]
From, e.g., Lemma~3.2 in~\cite{bk2013} we obtain
\[ 
\partial_x \widehat{x}^m(v,x) =
\frac{1/a_m \left( \widehat{x}^m(v,x) \right)}{\sum_{k=1}^M 1/a_k \left(
\widehat{x}^k(v,x) \right)}, \; m=1,\ldots,M,
\] 
which is uniformly bounded away from 0 and 1. Hence,
\[
\int_0^{\infty} \left| a_m\left( \widehat{x}^m(v, x) \right) - \underline{a}_m\right| \,dx 
= \int_{\widehat{x}^m(v,0)}^{\infty} \left| a_m(y) -  \underline{a}_m \right| \left( \partial_x \widehat{x}^m(v,y)
\right)^{-1} \,dy<\infty
\]
since $a_m$ satisfies condition~\eqref{eq:13}. Together with the above
Lipschitz estimate this yields that~\eqref{eq:13} holds for $a_r(v,x)$
and $\underline{a}_r$. The argument for stabilization at $-\infty$ is
the same.
\end{proof}

The second lemma will allow us to compare a utility function
satisfying condition~\eqref{eq:13} with a sum of two exponential
utilities:
\begin{Lemma}\label{lem:2}
  If a strictly concave, increasing utility function $u \in C^2(\RR)$
  with $\lim_{x\uparrow \infty} u(x)=0$ satisfies
  condition~\eqref{eq:13}, then there are constants $C_1, C_2>0$ such
  that
  \begin{equation}\label{eq:16}
  C_1 \left( - e^{-\underline{a} x} - e^{-\overline{a} x} \right) 
  \leq u(x) 
  \leq C_2 \left( -e^{-\underline{a} x} - e^{-\overline{a} x}\right) 
  , \quad x \in \mathbb{R}
\end{equation}
where $\underline{a}$ and $\overline{a}$ are the same constants as
in~\eqref{eq:13}.
\end{Lemma}

\begin{proof} 
First, we note that
\begin{equation} \label{eq:17} u(x) = - u'(0)
\int_x^{\infty} \exp \left( - \int_0^y a(z) dz \right) dy, \quad x \in
\mathbb{R}\,.
\end{equation}
Hence, by
L'Hopital's rule,
\begin{align*} 
\lim_{x \rightarrow \infty} \frac{u(x)}{-e^{-\underline{a} x}} 
&= \lim_{x \rightarrow \infty} \frac{u'(0) \int_x^{\infty} \exp \left(
- \int_0^y a(z) dz \right) dy}{e^{-\underline{a} x}} \\ 
& = \lim_{x \rightarrow \infty} 
 \frac{u'(0) \exp \left( - \int_0^x a(z) dz \right)}{\underline{a}
e^{-\underline{a} x}} \\ 
& = \frac{u'(0)}{\underline{a}}
\exp \left( \int_0^{\infty} \left( \underline{a} - a(z) \right) dz \right) \in (0,
\infty)
\end{align*} 
due to the integrability condition~\eqref{eq:13}. It follows that
\begin{equation} \label{eq:18} 
-c_2 e^{-\underline{a} x} \leq u(x) \leq -c_1 e^{-\underline{a} x}, \quad x > 0,
\end{equation} 
for positive constants $c_1$ and $c_2$. In conjunction with $0 <
\underline{a} \leq \overline{a}$ this yields
\begin{equation*} 
\limsup_{x \rightarrow \infty} \frac{u(x)}{-e^{-\underline{a} x} - e^{-\overline{a} x}} 
\leq \limsup_{x \rightarrow \infty}\frac{c_2 e^{-\underline{a} x}}{e^{-\underline{a} x} + e^{-\overline{a} x}} < \infty
\end{equation*}
 as well as
\begin{equation*} 
\liminf_{x \rightarrow \infty} \frac{u(x)}{-e^{-\underline{a} x} - e^{-\overline{a} x}} 
\geq \liminf_{x \rightarrow \infty}\frac{c_1 e^{-\underline{a} x}}{e^{-\underline{a} x} + e^{-\overline{a} x}} > 0.
\end{equation*}
This shows that there exists $C_2>0$ such that the right estimate
of~\eqref{eq:16} holds true. The argument for finding $C_1>0$ such
that the left estimate holds is completely analogous.
\end{proof}

The comparison with exponential utilities is also at the heart of the
following technical result:

\begin{Lemma} \label{lem:3} 
  Let $u \in C^2(\RR)$ be a strictly concave, increasing utility
  function with $\lim_{x \uparrow \infty} u(x)=0$ whose risk aversion
  is bounded away from zero and infinity and satisfies
  condition~\eqref{eq:13} for some constants $0<\underline{a} \leq
  \overline{a}<\infty$. Define $f(x) \set x -
  (-x)^{\underline{a}/\overline{a}}$ for $x \leq 0$.

  There exists a positive constant $C > 0$ depending only on the
  utility function $u$ such that
  \begin{align*} 
    \E \left[ u \left( x+\Sigma \right) | \mathcal{F}_t \right] \geq & 
     C f \left(
      \E \left[ u \left( x+\Sigma\right) | \mathcal{F}_{t-1} \right] \frac{\E
        \left[ \exp \left( - \overline{a} \Sigma \right) | \mathcal{F}_t
        \right]}{\E \left[ \exp \left( - \overline{a} \Sigma
          \right) | \mathcal{F}_{t-1} \right]} \right)
  \end{align*}
  for any $x \in \mathbb{R}$, any $t=1,\ldots,T$ and any random
  variable $\Sigma$ with finite exponential moments $\E
  \exp(-\overline{a} \Sigma)+\E \exp(-\underline{a} \Sigma)<\infty$.
\end{Lemma}

\begin{proof}
  Let $t = 1, 2, \ldots, T$ be arbitrary. For given $\Sigma$ with the
  above finite exponential moments and for any fixed $x \in
  \mathbb{R}$, the growth estimate~\eqref{eq:16} of Lemma~\ref{lem:2}
  allows us to define the $\mathcal{F}_{t-1}$-measurable random
  variable $\xi$ by
\begin{equation} \label{eq:19}
\E \left[ u \left( x+\Sigma \right) | \mathcal{F}_{t-1} \right] = C_2 \E \left[ - \exp \left( - \overline{a} \left( \xi+\Sigma \right) \right) | \mathcal{F}_{t-1} \right],
\end{equation}
where $C_2>0$ is any constant such that~\eqref{eq:16} holds. By the
same estimate we obtain 
\begin{align*}
C_2 \E \left[ - \exp \left( - \overline{a} \left( \xi+ \Sigma \right) \right) | \mathcal{F}_{t-1} \right] & = \E \left[ u \left( x+\Sigma \right) | \mathcal{F}_{t-1} \right] \\
& \leq C_2 \E \left[ - \exp \left( - \overline{a} \left(x+ \Sigma\right) \right) | \mathcal{F}_{t-1} \right],
\end{align*}
which implies that $\xi \leq x$. We use this in the other part of
estimate~\eqref{eq:16} to deduce
\begin{align*}
& \E \left[ u \left( x+\Sigma \right) | \mathcal{F}_t \right] \\
& \geq C_1 \E \left[ - \exp \left( - \overline{a} \left(x+ \Sigma \right) \right) - \exp \left( - \underline{a} \left(x+ \Sigma \right) \right) | \mathcal{F}_t \right] \\
& \geq  C_1 \E \left[ - \exp \left( - \overline{a} \left(\xi+ \Sigma\right) \right) - \exp \left( - \underline{a} \left(\xi+ \Sigma \right) \right) | \mathcal{F}_t \right] \\
& = C_1 \cbr{ - \exp \left( - \overline{a} \xi \right) \E \left[ \exp \left( - \overline{a} \Sigma \right) | \mathcal{F}_t \right] - \exp \left( - \underline{a} \xi \right) \E \left[ \exp \left( - \underline{a} \Sigma \right) | \mathcal{F}_t \right] } 
\end{align*}
By definition of $\xi$ we can write
\begin{equation*}
- \exp\left( -\overline{a} \xi \right) = \frac{1}{C_2} \frac{\E \left[ u \left(x+ \Sigma \right) | \mathcal{F}_{t-1} \right]}{\E \left[ \exp \left( -\overline{a} \Sigma \right) | \mathcal{F}_{t-1} \right]}.
\end{equation*}
This yields
\begin{align*}
\E \left[ u \left( x+\Sigma \right) | \mathcal{F}_t \right] 
\geq & \frac{C_1}{C_2} \E \left[ u \left(x+ \Sigma \right)| \mathcal{F}_{t-1} \right]
        \frac{\E \left[ \exp \left( - \overline{a} \Sigma \right) | \mathcal{F}_t \right]}
                {\E \left[ \exp \left( - \overline{a} \Sigma \right) | \mathcal{F}_{t-1} \right]} \\
& - C_1\left(
       -\frac{1}{C_2} \frac{\E \left[ u \left( x+\Sigma \right) | \mathcal{F}_{t-1} \right]}{\E \left[ \exp \left( -\overline{a} \Sigma \right) | \mathcal{F}_{t-1} \right]}\right)^{\underline{a}/\overline{a}} \E \left[ \exp \left( - \overline{a} \Sigma \right)^{\underline{a}/\overline{a}}  | \mathcal{F}_t \right] 
\end{align*}
and after applying Jensen's inequality with the concave function $x
\mapsto x^{\underline{a}/\overline{a}}$ we see that there is a
constant $C>0$ such that
\begin{align*}
\E \left[ u \left(x+\Sigma \right) | \mathcal{F}_t \right]
 \geq &
 C \Bigg\{  \E  \left[ u \left(x+ \Sigma \right) | \mathcal{F}_{t-1} \right] 
              \frac{\E \left[ \exp \left( - \overline{a} \Sigma \right) | \mathcal{F}_t
                \right]}
              {\E \left[ \exp \left( - \overline{a} \Sigma \right)| \mathcal{F}_{t-1} \right]}  \\
& \quad-  \left(- \E \left[ u \left( x+\Sigma \right)| \mathcal{F}_{t-1} \right] \frac{\E \left[ \exp \left( - \overline{a} \Sigma \right) | \mathcal{F}_t\right]}{\E \left[ \exp \left( - \overline{a} \Sigma \right) | \mathcal{F}_{t-1} \right]} \right)^{\underline{a}/\overline{a}} \Bigg\} \\
= & C f \left( \E \left[ u \left( x+\Sigma \right) | \mathcal{F}_{t-1} \right] \frac{\E \left[ \exp \left( - \overline{a} \Sigma \right) | \mathcal{F}_t\right]}{\E \left[ \exp \left( - \overline{a} \Sigma \right) | \mathcal{F}_{t-1} \right]} \right)
\end{align*}
where $f$ is as defined in the assertion of our lemma.
\end{proof}

We are finally in a position to give the

\paragraph{Proof of Theorem~\ref{thm:1}.} 
Let $t \in \cbr{1,\ldots,T}$ and $Q^n$, $u^n$, $n=1,2,\ldots$ be
as in Definition~\ref{def:1}. For notational simplicity let us denote
\[
U^n_s \set U^{t-1,u^n,Q^n}_s, \; X^n_s \set X^{t-1,u^n,Q^n}_s, \; V^n_s
\set V^{t-1,u^n,Q^n}_s, \quad s=t,\ldots,T\,.
\]
We first recall that by Theorem~4.2 in \cite{bk2013} $G_{t-1}$
is contained in $\widetilde{G}^2(c)$, a class of saddle functions
introduced there. Property~(G7) of these special saddle functions
amounts in our context to
\begin{equation}
\label{eq:20}
\frac{1}{c} \leq -u^{m,n} V^{m,n}_t \leq c
\end{equation}
where $c$ is the bound on the market makers' risk aversions occurring
in~\eqref{eq:1}. Thus, since by assumption $(-u^{m,n})_{n=1,2,\ldots}$ is bounded away from
zero and $\infty$, so are the weights $V^{m,n}_t$, $n=1,2,\ldots$ for any
$m=1,\ldots,M$. As a consequence $\underline{V}^m_t \set \inf_{n=1,2,\ldots}
V^{m,n}_t>0$ and $\overline{V}^m_t \set \sup_{n=1,2,\ldots}
V^{m,n}_t<\infty$ yield finite $\cF_{t-1}$-measurable bounds on the
initial weight of each market maker $m=1,\ldots,M$.

Let us next argue that
\begin{align}
  0 \geq &\abr{\underline{V}_t,U^n_t} \geq \abr{V^n_t,U^n_t}=
  \condexp{r(V^n_t,\Sigma(X^n_t,Q^n_t))} \label{eq:21}\\
  \geq &
  \condexp{r(\overline{V}_t,\Sigma(X^n_t,Q^n_t))} \label{eq:22}\\
  \geq & C_{\overline{V}_t}
  f\Bigg(\condexp[\cF_{t-1}]{r(\overline{V}_t,\Sigma(X^n_t,Q^n_t))}\label{eq:23}\\\nonumber
  &\qquad\qquad \cdot
  \frac{\condexp{\exp(-\overline{a}\Sigma(0,Q^n_t))}}{\condexp[\cF_{t-1}]{\exp(-\overline{a}\Sigma(0,Q^n_t))}}\Bigg)
\end{align}
for some random variable $C_{\overline{V}_t}>0$. Indeed, the estimates
in~\eqref{eq:21} are immediate from $0<\underline{V}_t^m \leq
V^{m,n}_t$ and $U^{m,n}_t <0$, $m=1,\dots,M$. The identity
in~\eqref{eq:21} holds because, by~\eqref{eq:6}, $U^n_t$ is the vector
of $\cF_t$-conditional expected utilities that our market makers
obtain when at time $t-1$ the Pareto allocation of
$\Sigma(X^n_t,Q^n_t)$ is formed given the weights
$V^n_t$. Estimate~\eqref{eq:22} follows because the representative
agent's utility function $r(v,x)$ is decreasing in the weights
$v$. Finally, \eqref{eq:23} follows because our Lemmas~\ref{lem:1}
and~\ref{lem:2} above allow us to apply Lemma~\ref{lem:3} to the
(random) utility function $u\set r(\overline{V}_t,.)$ which provides
us with the required random variable $C_{\overline{V}_t}>0$.

Recalling that $\Sigma_0 = \widetilde{\Sigma}_0+\abr{q_0,\psi}$ with
$q_0 \in \RR^J$ and bounded $\widetilde{\Sigma}_0$, we have the
following estimate for the
ratio in~\eqref{eq:23}:
\[
\frac{\condexp{\exp(-\overline{a}\Sigma(0,Q^n_t))}}{\condexp[\cF_{t-1}]{\exp(-\overline{a}\Sigma(0,Q^n_t))}}
\leq e^{2\overline{a}\|\widetilde{\Sigma}_0\|_{\infty}} \frac{\condexp{\exp(-\abr{\overline{a}(q_0+Q^n_t),\psi})}}
{\condexp[\cF_{t-1}]{\exp(-\abr{\overline{a}(q_0+Q^n_t),\psi})}}\,.
\]
By our assumption of decreasing exponential tails the latter ratio
and, thus, also the former ratio converge to zero on $D_t \set
\cbr{\nu_{t-1} \prec \nu_t} \cap \cbr{\lim_n |Q^n_t|=\infty}$. This is
a set with positive probability
\begin{align*}
 \mathbb{P} \left[ D_t \right] = \E \left[ \mathbf{1}_{\cbr{\lim_n |Q^n_t|=\infty}} \condprob[\cF_{t-1}]{\nu_{t-1} \prec \nu_t} \right]\,,
\end{align*}
which is strictly positive by \eqref{eq:12}.
We shall argue below that
\begin{equation}\label{eq:22b}
0 \geq \condexp[\cF_{t-1}]{r(\overline{V}_t,\Sigma(X^n_t,Q^n_t))} \geq
\abr{\overline{V}_t,u^n}, \; n=1,2,\ldots,
\end{equation}
and so the first conditional expectation in~\eqref{eq:23} is bounded
in $n$ by assumption on $(u^n)_{n=1,2,\ldots}$.  As a consequence, the
right side in~\eqref{eq:23} converges to $f(0)=0$ on $D_t$ when $n
\uparrow \infty$. This implies $\abr{\underline{V}_t,U^n_t} \to 0$ on
this set. Since $\underline{V}_t^m>0$, $m=1,\ldots,M$, this yields
that actually $U^n_t \to 0$ on $D_t$. Because $U^n$ is a martingale,
we have $U^n_t = \condexp{U^n_T}$ and so $U^n_T \to 0$ in probability
on $D_t$. This, however, is equivalent to $X^n_T+\abr{Q^n_T,\psi} =
\sum_{m=1}^M u_m^{-1}(U^{m,n}_T)-\Sigma_0 \to \infty$ in probability
on $D_t$ which establishes the asserted efficient friction.

It remains to verify~\eqref{eq:22b}. For this, note that the Pareto
allocation $\alpha^n_t$ of $\Sigma(X_t^n,Q^n_t)$ with weights $V^n_t$
gives us
\[
r(\overline{V}_t,\Sigma(X_t^n,Q^n_t)) \geq \sum_{m=1}^M \overline{V}^m_t
u_m(\alpha^{m,n}_t), \; n=1,2,\ldots\,.
\]
Our model's utility indifference principle~\eqref{eq:3} yields
\[
\condexp[\cF_{t-1}]{u_m(\alpha^{m,n}_t)} = U^{m,n}_{t-1}=u^{m,n}, \;
m=1,\ldots,M, \; n=1,2,\ldots,
\]
and so~\eqref{eq:22b} follows by taking an $\cF_{t-1}$-conditional
expectation in the preceding estimate and recalling that
$\overline{V}_t$ is $\cF_{t-1}$-measurable.\qed

\section{Ramifications}

In this section we collect a few supplementary results and
illustrations.  We first illustrate in
Section~\ref{sec:counterexample} that our tail condition on $\psi$,
\eqref{eq:12}, is necessary for the existence of optimal
superreplicating strategies even in a simple binomial model with two
periods and one exponential market maker.  We furthermore show in
Section~\ref{sec:stringentassumption} that under more stringent
assumptions on $\psi$, efficient friction holds even without the
requirement on the market makers' risk
aversions. Section~\ref{sec:binomial} is finally concerned with the
special case of a general multi-period binomial model.

\subsection{A binomial model where the superreplication price is
  not attained}\label{sec:counterexample}

We consider a two-period model with one asset and one market maker,
where $\Omega \set \cbr{-1,+1}^2$. For $\omega = (y_1, y_2) \in \Omega$,
let $Y_t(\omega) \set y_t$ be the projection of $\omega$ to its $t$-th
component, $t = 1,2$. The filtration $(\mathcal{F}_t)_{t=0,1,2}$ is
generated by $Y = (Y_t)_{t=1,2}$. The distribution of $Y$ is determined by
\begin{equation*}
\P[ Y_1 = +1]  \set p_1, \quad 
\P[ Y_2 = +1 \,|\, Y_1 = +1] \set p_2, \quad 
\P[ Y_2 = +1\, |\, Y_1 = -1] \set p_3,
\end{equation*}
where $p_1, p_2, p_3 \in (0, 1)$ with $p_2 \neq p_3$. 
Moreover, the single market maker's utility function is given by $u(x)
\set -e^{-\alpha x}$ for $\alpha > 0$, the initial endowment is $\Sigma_0 \set 0$ and $\psi$
is determined by
\begin{equation*}
\psi \left( Y_1, +1 \right) \set \psi^u , \quad \psi\left( Y_1, -1 \right) \set \psi^d \text{ with } \psi^u > \psi^d.
\end{equation*}

The specific form of exponential utility and predictability of the cash balance $X^Q$ of the strategy $Q$ yield that
\begin{equation*}
-1 = u(0) = \E \left[ u\left( X_1^Q + Q_1 \psi \right) \right] = -\E \left[ e^{-\alpha(X_1^Q + Q_1 \psi)} \right] = - e^{-\alpha X_1^Q} \E \left[ e^{- \alpha Q_1 \psi} \right].
\end{equation*}
Thus
\begin{equation*}
X_1^Q = \frac{1}{\alpha} \log \left( \E \left[ e^{- \alpha Q_1 \psi} \right] \right)
\end{equation*}
and the utility level $U_1^Q$ at time $1$ of the strategy $Q$ satisfies
\begin{equation*}
U_1^Q = -\frac{\E \left[ e^{-\alpha Q_1\psi} |Ê\mathcal{F}_1 \right]}{\E \left[ e^{- \alpha Q_1\psi} \right]}.
\end{equation*}
By direct calculation, we find that 
\begin{equation} \label{e.u1q}
U_1^Q = - \frac{A e^{- \alpha Q_1 (\psi^u -\psi^d)} + (1-A)}{B e^{-\alpha Q_1 (\psi^u - \psi^d)} + (1-B)}  \text{
  with } A \set\left\{ \begin{array}{cc} p_2 & \text{ if } Y_1 = +1 \\
p_3 & \text{ if } Y_1 = -1 \end{array} \right. 
\end{equation}
and $B \set p_1p_2 + (1-p_1)p_3$ in our model.

We consider taking ever larger positions at the first time
step and liquidating our positions in the second time step. So let
$(Q_1^n)_{n \in \mathbb{N}}$ be a sequence such that $Q_1^n
\rightarrow \infty$ as $n \rightarrow \infty$. Then
\begin{equation*}
U_1^n \set U_1^{Q^n} \rightarrow U_1 = -\frac{1-A}{1-B} 
\text{ as }  n \rightarrow \infty.
\end{equation*}
Since we liquidate our position after the first period, $Q_2^n = 0$
for every $n =1,2,\ldots$. This amounts to $U_1^n = - e^{-\alpha X_2^n}$ and thus
\begin{equation*}
\lim_n X^n_2 = X_2 \set \frac{1}{\alpha} \log \left( -
  \frac{1}{U_1} \right) =  \frac{1}{\alpha} \log \left( \frac{1-B}{1-A} \right),
\end{equation*}
because of predictability of $X$. 

Next we shall show that the superreplicating cost of the claim $H \set
-X_2 \in \mathcal{F}_1$ is $\pi^H = 0$ and that there exists no
strategy $Q$ that superreplicates with initial capital $\pi^H =
0$. Let us first show $\pi^H \leq 0$. Since $-X_2^n \rightarrow H$,
there exists a sequence of real numbers $\epsilon_n \downarrow 0$ as
$n \uparrow \infty$ and
\begin{equation*}
\epsilon_n - X_2^n \geq H, \; n=1,2,\ldots\,.
\end{equation*}
Moreover, the strategy $\left( Q_t^n \right)_{t = 1,2} = (Q_1^n, 0)$
yields the cash balance $X_2^n$ for all $n = 1.2.\ldots$. Therefore,
$\epsilon_n \geq \pi^H$ for every such $n$ and by sending $\epsilon_n
\downarrow 0$, we obtain that $\pi^H \leq 0$.  Thus if $\pi^H$ is not
zero, then there exists $\epsilon > 0$ and a superreplicating strategy
$\widetilde{Q}$ with cash balance $\widetilde{X}$ such that
\begin{equation*}
- \epsilon - \widetilde{X}_2 - \widetilde{Q}_2\psi \geq H\,.
\end{equation*}
Hence,
\begin{equation*} 
\widetilde{X}_2 + \widetilde{Q}_2\psi < \epsilon + \widetilde{X}_2 + \widetilde{Q}_2\psi \leq -H\,. 
\end{equation*}
Moreover, by construction of $H$, we have $\E \left[u(-H)\right] =
u(0)$. These two observations lead to the contradiction
\begin{equation*}
u(0) = \E \left[ u \left( \widetilde{X}_2 + \widetilde{Q}_2\psi \right) \right] < \E\left[ u(-H) \right] = u(0).
\end{equation*}

It remains to show that $\pi^H=0$ is not attained by any strategy
$\widetilde{Q}$. If $\widetilde{Q}$ is a strategy that superreplicates with
initial capital $\pi^H = 0$ and cash balance $\widetilde{X}$, i.e.
\begin{equation*}
- \widetilde{X}_2 - \widetilde{Q}_2\psi \geq H \,,
\end{equation*}
then 
\begin{equation*}
\E \left[ u \left( \widetilde{X}_2 + \widetilde{Q}_2\psi \right)  \right] = u(0) = \E \left[ u(-H )\right] 
\end{equation*}
yields that $ \widetilde{X}_2 + \widetilde{Q}_2\psi = - H = X_2$. Then since $X_2$ is $\mathcal{F}_1$ measurable,
\begin{equation*}
\E \left[ u \left( \widetilde{X}_2 + \widetilde{Q}_2\psi \right) | \mathcal{F}_1 \right] = \E \left[ u(X_2) | \mathcal{F}_1 \right] = u(X_2).
\end{equation*}
So $u(X_2)$ is the utility level at time $1$ of the strategy
$\widetilde{Q}$, i.e. $U_1^{\widetilde{Q}} = u(X_2)$. However, by
differentiating equation \eqref{e.u1q} we obtain
\begin{equation*}
\frac{\partial}{\partial Q_1} U_1^{Q} = \alpha \left( \psi^u - \psi^d \right) e^{-\alpha Q_1 \left( \psi^u - \psi^d \right)} \frac{A-B}{\left(Be^{-\alpha Q_1 \left( \psi^u - \psi^d \right)} + (1-B) \right)^2},
\end{equation*}
where $A$ and $B$ are given as in \eqref{e.u1q}. Without loss of generality assume 
that $p_2 > p_3$. Then, $U_1^Q$ is
strictly increasing in $Q_1$ on the set $\cbr{ Y_1 = +1 }$ and strictly
decreasing on $\cbr{ Y_1 = -1}$. Since the utility level $u(X_2)$ is the
limiting value of positions $Q_1^n$ tending to $+\infty$ and $U_1^Q$
is strictly monotone, there cannot exist a finite $\widetilde{Q}_1$ with
utility level $u(X_2)$, a contradiction.

\subsection{Efficient friction when extremal payoffs may change any
  time} \label{sec:stringentassumption}

In this section we will consider the special case where only one
security is marketed and so $\psi$ is a real-valued random variable.
Let us denote by $\underline{\psi}_{t}$ its $\mathcal{F}_t$-measurable
essential infimum of $\psi$ and by $\overline{\psi}^{t}$ its
$\mathcal{F}_t$-measurable essential supremum, i.e.
\begin{align*}
\underline{\psi}_{t}  \set \esssup \cbr{ \zeta \in \mathcal{F}_t :
  \zeta \leq \psi } \text{ and }
\overline{\psi}^{t} \set \essinf \cbr{ \zeta \in \mathcal{F}_t : \zeta \geq \psi }.
\end{align*}

\begin{Proposition} \label{prop:1}
A single security $\psi \in L^0(\RR)$ exhibits decreasing exponential tails in the sense
of condition~\eqref{eq:12} if its conditional infima and suprema are potentially
strictly monotone, i.e., if
\begin{equation} \label{eq:24}
\condprob[\cF_{t-1}]{ \underline{\psi}_{t-1} < \underline{\psi}_{t}, \, \overline{\psi}^{t-1} > \overline{\psi}^{t}    } >0, \quad t = 1, \ldots, T.
\end{equation}
Moreover, if $|\Omega|<\infty$ or, more generally, if
\begin{equation}\label{eq:25} 
\condprob{
  \psi = \underline{\psi}_{t} } > 0 \text{ and } \condprob{\psi =
  \overline{\psi}^{t} } > 0, \;t = 0, \ldots, T,
\end{equation} 
then conditions~\eqref{eq:12} and~\eqref{eq:24} are even equivalent.
\end{Proposition}
\begin{proof}
Since
\begin{equation*}
A \set\cbr{ \underline{\psi}_{t-1} < \underline{\psi}_{t}, \ \overline{\psi}^{t-1} > \overline{\psi}^{t} }  = \bigcup_{\epsilon > 0} \cbr{ \underline{\psi}_{t-1} + \epsilon < \underline{\psi}_{t}, \ \overline{\psi}^{t-1} - \epsilon > \overline{\psi}^{t} } 
\end{equation*}
the implication `\eqref{eq:24} $\Longrightarrow$ \eqref{eq:12}'
will follow if we show
\begin{equation*}
\cbr{ \underline{\psi}_{t-1} + \epsilon < \underline{\psi}_{t}, \ \overline{\psi}^{t-1} - \epsilon > \overline{\psi}^{t} } \subseteq \cbr{ \lim_{|q| \rightarrow \infty} \frac{\E \left[ e^{q \psi} | \cF_t\right]}{\E \left[ e^{q \psi} | \cF_{t-1}\right]} = 0 }
\end{equation*}
 for any $\epsilon > 0$. To this end, note that because $\overline{\psi}^{t-1}$ is $\mathcal{F}_{t-1}$-measurable, we obtain
the following estimate for any $q>0$ :
\begin{align*}
\E \left[ e^{q \psi} |Ê\cF_{t-1} \right] & \geq \E \left[ e^{q \psi} \mathbf{1}_{\cbr{ \psi > \overline{\psi}^{t-1} - \epsilon }} \big| \cF_{t-1} \right] \\
& \geq \exp \left( q \left( \overline{\psi}^{t-1} - \epsilon \right) \right) \P \left( \psi > \overline{\psi}^{t-1} - \epsilon | \mathcal{F}_{t-1} \right).
\end{align*}
Now $\P \left( \psi > \overline{\psi}^{t-1} - \epsilon |
  \mathcal{F}_{t-1} \right) > 0$ by definition of the
conditional essential supremum $\overline{\psi}^t$. 

So
\begin{equation*}
  0 \leq \lim_{q \rightarrow \infty} \frac{\E \left[ e^{q \psi} | \cF_t \right]}{\E \left[ e^{q \psi} | \cF_{t-1} \right]} \leq \lim_{q \rightarrow \infty} \frac{ \exp \left( q \left(\overline{\psi}^{t} - \overline{\psi}^{t-1} + \epsilon \right) \right)}{\P \left( \psi > \overline{\psi}^{t-1} - \epsilon | \mathcal{F}_{t-1} \right)} = 0
\end{equation*}
on the set $\cbr{ \underline{\psi}_{t-1} + \epsilon < \underline{\psi}_{t}, \ \overline{\psi}^{t-1} - \epsilon > \overline{\psi}^{t} }$.\\

Similarly for $q <0$, we use the estimate
\begin{align*}
\E \left[ e^{q \psi} | \cF_{t-1} \right] & \geq \E \left[ e^{q \psi} \mathbf{1}_{\cbr{ \psi < \underline{\psi}_{t-1} + \epsilon }} | \cF_{t-1} \right] \\
& \geq \exp \left( q \left( \underline{\psi}_{t-1} + \epsilon \right) \right) \P \left( \psi < \underline{\psi}_{t-1} + \epsilon | \mathcal{F}_{t-1} \right)
\end{align*}
to deduce that 
\begin{equation*}
0 \leq \lim_{q \rightarrow -\infty} \frac{\E \left[ e^{q \psi} |Ê\cF_t \right]}{\E \left[ e^{q \psi} |Ê\cF_{t-1} \right]} \leq \lim_{q \rightarrow -\infty} \frac{ \exp \left( q \left( \underline{\psi}_{t} - \underline{\psi}_{t-1} - \epsilon \right) \right)}{\P \left( \psi < \underline{\psi}_{t-1} + \epsilon | \mathcal{F}_{t-1} \right)} = 0
\end{equation*}
on the set $\cbr{ \underline{\psi}_{t-1} + \epsilon < \underline{\psi}_{t}, \ \overline{\psi}^{t-1} - \epsilon > \overline{\psi}^{t} }$, because $\P \left( \psi < \underline{\psi}_{t-1} + \epsilon | \mathcal{F}_{t-1} \right) > 0$. This establishes the first claim.\\

For the second claim, we note that given~\eqref{eq:25} holds we
obtain by dominated convergence
\begin{align*}
\lim_{q \downarrow -\infty} &\frac{\condexp[\cF_{t}]{ e^{q \psi}  }}{\condexp[\cF_{t-1}]{ e^{q \psi} } } = \lim_{q \downarrow - \infty} \frac{\condexp{ e^{q \psi} \left( \mathbf{1}_{\cbr{ \psi = \underline{\psi}_{t-1} }} + \mathbf{1}_{\cbr{ \psi > \underline{\psi}_{t-1} }}  \right) }}{\condexp[\cF_{t-1}]{ e^{q \psi} \left( \mathbf{1}_{\cbr{ \psi = \underline{\psi}_{t-1} }} + \mathbf{1}_{\cbr{ \psi > \underline{\psi}_{t-1} }}  \right)}}\\
& = \lim_{q \downarrow - \infty} \frac{e^{q \underline{\psi}_{t-1}} \cbr{ \condprob[\cF_{t}]{\psi = \underline{\psi}_{t-1} } + \condexp[\cF_{t}]{ e^{q (\psi - \underline{\psi}_{t-1})} \mathbf{1}_{\cbr{\psi > \underline{\psi}_{t-1} }} } }}{e^{q \underline{\psi}_{t-1}} \cbr{ \condprob[\cF_{t-1}]{\psi = \underline{\psi}_{t-1} } + \condexp[\cF_{t-1}]{ e^{q (\psi - \underline{\psi}_{t-1})} \mathbf{1}_{\cbr{\psi > \underline{\psi}_{t-1} }} } }} \\
& = \frac{\condprob[\cF_{t}]{ \psi = \underline{\psi}_{t-1}}}{\condprob[\cF_{t-1}]{ \psi = \underline{\psi}_{t-1}}} 
\end{align*}
and, similarly,
\[
\lim_{q \uparrow +\infty} \frac{\condexp[\cF_{t}]{ e^{q \psi}  }}{\condexp[\cF_{t-1}]{ e^{q \psi} } } = \frac{\condprob[\cF_{t}]{ \psi = \overline{\psi}^{t-1}}}{\condprob[\cF_{t-1}]{ \psi = \overline{\psi}^{t-1}}} \,.
\]
So, if the conditional supremum and infimum both change from time $t-1$
to $t$ with positive probability, the above ratios of conditional
probabilities are zero along the limits as $q\downarrow -\infty$
and $q\uparrow +\infty$. This was to be shown.
\end{proof}

Our next result illustrates that our assumption of stabilizing
asymptotic risk aversions formulated in Theorem~\ref{thm:1} is not
needed for efficient friction to hold when essential suprema and
infima are strictly monotone:

\begin{Theorem}\label{thm:2}
Under Assumptions~\ref{asp:1} and~\ref{asp:2}, a model with a single marketed security
$\psi$ satisfying condition~\eqref{eq:24} exhibits efficient friction.
\end{Theorem}

\begin{Remark}
As an application of this theorem, all non-degenerate trinomial and higher monomial
models exhibit efficient friction.
\end{Remark}

\begin{proof}
We start with the same observations as in the proof of
Theorem~\ref{thm:1} and note that up to and including the
estimate~\eqref{eq:22} all arguments hold true under the assumptions
of the present theorem.

It thus suffices to identify, for $t=1,\dots,T$, a set $D_t \in \cF_t$
with positive probability where
$\condexp{r(\overline{V}_t,\Sigma(X^n_t,Q^n_t))} \to 0$. It then
follows that $\abr{\underline{V}_t,U^n_t} \to 0$ on this set and, just
as in the proof of Theorem~\ref{thm:1}, we can conclude that $U^n_t
\to 0$ and, thus, also $U^n_T \to 0$ in probability on $D_t$. As
before this is equivalent to $X^n_T+ Q^n_T \psi \to \infty$ in
probability on $D_t$, proving the asserted efficient friction.

By working with subsequences we can confine ourselves to the case
where $\cbr{\lim_n Q^n_t = \pm \infty}$ has positive probability. The
argument on $\cbr{\lim_n Q^n_t = -\infty}$ being similar, let us assume
that the set $\cbr{\lim_{n} Q^n_t = +\infty}$ has positive
probability. Clearly, on $\cbr{Q^n_t>0}$ we have
\begin{align*}
0 &\geq\condexp{r(\overline{V}_t,\Sigma(X^n_t,Q^n_t))}  \\
& \geq \condexp{r(\overline{V}_t,\Sigma_0 + Q^n_t(\underline{\psi}_t+X^n_t/Q^n_t))}. 
\end{align*}
By condition~\eqref{eq:24} we can find an $\epsilon>0$ such that $D_t
\set \cbr{\lim_{n} Q^n_t = +\infty} \cap
\cbr{\underline{\psi}_{t-1}+\epsilon<\underline{\psi}_{t},
\overline{\psi}^{t-1}-\epsilon>\overline{\psi}^t}$ has positive
probability, since $\cbr{\lim_{n} Q^n_t = +\infty}$ is $\cF_{t-1}$-measurable
and
$$
D_t \uparrow \cbr{\lim_{n} Q^n_t = +\infty} \cap 
\cbr{\underline{\psi}_{t-1} <\underline{\psi}_{t},
\overline{\psi}^{t-1}>\overline{\psi}^t}
$$ 
as $\epsilon \downarrow 0$. We shall argue in Lemma~\ref{lem:4} below that
\begin{equation*}
  \lim_n X^n_t/Q^n_t = -\underline{\psi}_{t-1} \text{ on } \cbr{\lim_{n} Q^n_t = +\infty}\,.
\end{equation*} 
On $D_t$ we can thus furthermore find a random variable
$N^{\epsilon}$ such that for $n>N^{\epsilon}$ the above estimation
can be continued by
\begin{align*}
\dots \geq  & \condexp{r(\overline{V}_t,\Sigma_0 + Q^n_t(\underline{\psi}_t+X^n_t/Q^n_t))}\\
\geq & \condexp{r(\overline{V}_t,\Sigma_0 + Q^n_t(\underline{\psi}_t-\underline{\psi}_{t-1}-\epsilon/2))}\\
\geq & \condexp{r(\overline{V}_t,\Sigma_0 + Q^n_t\epsilon/2)}.
\end{align*}
On $D_t$ the latter expression converges to zero by dominated convergence as required.
\end{proof}

The following lemma establishes the asymptotics of cash balances for
extreme long and short positions in $\psi$:

\begin{Lemma}\label{lem:4}
Under Assumptions~\ref{asp:1} and~\ref{asp:2} we have
\begin{equation}
  \label{eq:26}
  \lim_n X^n_t/Q^n_t = -\underline{\psi}_{t-1} \text{ on } \cbr{\lim_{n} Q^n_t = +\infty}
\end{equation}
and
\begin{equation}
  \label{eq:27}
  \lim_n X^n_t/Q^n_t = -\overline{\psi}^{t-1} \text{ on } \cbr{\lim_{n} Q^n_t = -\infty}
\end{equation}
for any sequence of strategies $(Q^n)_{n=1,2,\dots}$ with cash
balances $(X^n)_{n=1,2,\dots}$ such that $(U^n_{t-1})_{n=1,2,\dots}$
is bounded away from $0$ and $-\infty$.
\end{Lemma}
\begin{proof}
The argument for~\eqref{eq:27} being similar, let us
establish~\eqref{eq:26}. To see that `$\geq$' holds, we note that
\begin{align*}
\abr{V^n_t,U^n_{t-1}} = \condexp[\cF_{t-1}]{\abr{V^n_t,U^n_t}} = \condexp[\cF_{t-1}]{r(V^n_t,\Sigma(X^n_t,Q^n_t))}
\end{align*}
is bounded by assumption on $u^n \set U^n_{t-1}$ because~\eqref{eq:20} holds for our
choice of $V^n_t$, $n=1,2,\dots$. Moreover, the last term is easily
estimated for any $A \in \cF_T$:
\begin{align*}
\condexp[\cF_{t-1}]{r(V^n_t,\Sigma(X^n_t,Q^n_t))}
& \leq \condexp[\cF_{t-1}]{r(\underline{V}_t,\Sigma(X^n_t,Q^n_t))1_A}\\
& \leq \condexp[\cF_{t-1}]{(r(\underline{V}_t,\Sigma_0) +
  \partial_x r(\underline{V}_t,\Sigma_0)(X^n_t+Q^n_t \psi))1_A}.
\end{align*}
On $\cbr{\lim_n Q^n_t = +\infty}$ we thus can divide by the
$\cF_{t-1}$-measurable quantity $Q^n_t$ in this series of inequalities
and let $n\uparrow \infty$ to deduce that
\[
0 \leq \condexp[\cF_{t-1}]{
  \partial_x r(\underline{V}_t,\Sigma_0)(\liminf_n X^n_t/Q^n_t+ \psi)1_A}.
\]
As $A \in \mathcal{F}_T$ is arbitrary, this implies that 
\begin{equation*}
\liminf_n X^n_t/Q^n_t+ \psi \geq 0,
\end{equation*}
i.e., $-\underline{\psi}_{t-1} \leq \liminf_n X^n_t/Q^n_t$, because
$X_t^n$, $Q_t^n$ are $\mathcal{F}_{t-1}$-measurable for all $n = 1, 2,
\ldots$.

On the other hand, 
\begin{align*}
\abr{V^n_t,U^n_{t-1}} = \condexp[\cF_{t-1}]{r(V^n_t,\Sigma(X^n_t,Q^n_t))}
\end{align*}
is also bounded away from zero, again by assumption on
$(u^n)_{n=1,2,\dots}$ and~\eqref{eq:20}. Moreover, on $\cbr{Q^n_t>0}$,
\begin{align*}
\condexp[\cF_{t-1}]{r(V^n_t,\Sigma(X^n_t,Q^n_t))}&\geq
\condexp[\cF_{t-1}]{r(\underline{V}_t,\Sigma_0+X^n_t+Q^n_t\underline{\psi}_{t-1})}.
\end{align*}
Because $X^n_t$ and $Q^n_t$ are $\cF_{t-1}$-measurable, $n=1,2,\dots$,
this implies that $\sup_n \cbr{X^n_t+Q^n_t
  \underline{\psi}_{t-1}}<\infty$. On $\cbr{\lim_n Q^n_t=+\infty}$ this
yields $\limsup_n X^n_t/Q^n_t+\underline{\psi}_{t-1} \leq 0$, proving
`$\leq$' in~\eqref{eq:26}.
\end{proof}

\subsection{(In)completeness of binomial models with exponential market makers}
\label{sec:binomial}

Next we consider a model with time horizon $T$, one market maker $M = 1$
with exponential utility $u(x)=-e^{- \alpha x}$, $x \in \RR$, one asset
$J=1$ and initial endowment $\Sigma_0$ satisfying $\E \left[ e^{-\alpha \Sigma_0} \right]  < \infty$
for the market maker. The model has a binomial structure, i.e., $\Omega = \cbr{-1,+1}^T$,
the filtration $(\cF_t)_{t=1,\ldots,T}$ is generated by
$Y_t(\omega)\set y_t$ for $\omega=(y_1,\ldots,y_T) \in \Omega$,
$t=1,\ldots,T$, and we assume $\P[\cbr{\omega}]>0$ for all $\omega \in
\Omega$.  The payoff $\psi$ can be any real-valued $\cF_T$-measurable
random variable. 

\begin{Theorem} \label{thm:3} For a market maker with
  exponential utility, a binomial model as described above is complete
  if and only if at any time we see a new best lower bound for $\psi$ in one
  possible evolution to the next time period and a new best upper
  bound for $\psi$ in the other:
 \begin{equation}
   \label{eq:28}
  \Omega = \cbr{\underline{\psi}_{t-1}<\underline{\psi}_{t}} \cup
  \cbr{\overline{\psi}^{t-1}>\overline{\psi}^{t}}, \; t=1,\ldots,T\,.
    \end{equation}
\end{Theorem}
\begin{proof} 
  By passing to $\P'$ with $d\P'/d\P = e^{-\alpha \Sigma_0}/\E e^{-\alpha
    \Sigma_0}$ we can assume $\Sigma_0=0$ without loss of
  generality. We will proceed through a number auxiliary claims:
\begin{description}
\item[Claim 1:] $H \in \mathcal{F}_T$ is attainable by a strategy $Q$
  with initial capital $\pi^H$ if and only if
\begin{equation} \label{eq:29}
\pi^H = \frac{1}{\alpha} \log \left( \E \left[ e^{\alpha H} \right] \right)
\end{equation}
and $Q$ is a predictable process that satisfies
\begin{equation} \label{eq:30} 
\frac{\E \left[ e^{\alpha H} | \cF_t \right]}{\E \left[ e^{\alpha
      H} | \cF_{t-1} \right]} 
= \frac{\E \left[ e^{-\alpha Q_t \psi} | \cF_t \right]}{\E \left[ e^{-\alpha Q_t
        \psi} | \cF_{t-1} \right]}, \text{ on } \cbr{Y_t=+1}, \  t = 1, 2, \ldots, T.
\end{equation}

Indeed, if $H$ is attainable by a strategy $Q$ with initial cost
$\pi^H$, then
\begin{equation*}
-1 = u(0) = \E \left[ u\left( X_T + Q_T \psi \right) \right] 
= \E \left[ u(\pi^H - H) \right] 
= -e^{-\alpha \pi^H} \E \left[ e^{\alpha H} \right],
\end{equation*}
since $H =\pi - X_T - Q_T \psi$. Therefore, the price $\pi^H$ and the
utility process $U$ satisfy
\begin{align*}
\pi^H & = \frac{1}{\alpha} \log \left( \E \left[ e^{\alpha H} \right] \right) \\ 
 \frac{U_t}{U_{t-1}} & = \frac{\E \left[ u\left(\pi^H - H\right ) | \cF_t \right] }{\E \left[ u\left(\pi^H - H\right ) | \cF_{t-1} \right]} = \frac{\E \left[ e^{\alpha H} | \cF_t \right]}{\E \left[ e^{\alpha H} | \cF_{t-1} \right]},  \; t = 1,2, \ldots, T.
\end{align*}
On the other hand, by predictability of the associated cash process $X$,
\begin{equation*}
U_{t-1} = \E \left[ u \left( X_t + Q_t \psi \right) | \cF_{t-1} \right] =
-e^{-\alpha X_t} \E \left[ e^{-\alpha Q_t \psi} | \cF_{t-1} \right], \ t =1, 2, \ldots, T
\end{equation*}
which implies that
\begin{equation*}
\frac{U_t}{U_{t-1}} = \frac{\E \left[ e^{-\alpha Q_t \psi} | \cF_t \right]}{\E \left[ e^{-\alpha Q_t \psi} | \cF_{t-1} \right]}, \quad \ t =1, 2, \ldots, T.
\end{equation*}
Hence, the strategy $Q$ satisfies~\eqref{eq:30}.

Conversely, let $\pi^H$ be given by~\eqref{eq:29} and $Q$ be
predictable with~\eqref{eq:30}. Then, in fact, the identity
in~\eqref{eq:30} holds on all of $\Omega=\cbr{Y_t=+1}\cup\cbr{Y_t=-1}$
because the random variable on either side has conditional expectation
1 given $\cF_{t-1}$. Now let $U_0 = u(0)$ and let $U$ be the process
defined recursively by
\begin{equation*}
\frac{U_t}{U_{t-1}} = 
\frac{\E \left[ e^{\alpha H} | \cF_t \right]}{\E \left[ e^{\alpha H} | \cF_{t-1}
  \right]} 
= \frac{\E \left[ e^{-\alpha Q_t \psi} | \cF_t \right]}{\E \left[ e^{-\alpha Q_t \psi} | \cF_{t-1} \right]}\, \quad  t = 1, 2, \ldots, T.
\end{equation*}
Then according to~\eqref{eq:29}
\begin{equation*}
\frac{U_T}{U_0} = \prod_{s=1}^T \frac{U_s}{U_{s-1}} = \frac{e^{\alpha H}}{\E \left[ e^{\alpha H} \right]} = e^{-\alpha(\pi^H-H)}.
\end{equation*}
On the other hand, $X$ given by
\begin{equation*}
X_t = \frac{1}{\alpha} \log \left( -\frac{\E \left[ e^{-\alpha Q_t \psi} | \cF_{t-1} \right]}{U_{t-1}} \right), \quad t =1, 2, \ldots, T,
\end{equation*}
is the cash process of $Q$ with indirect utility $U$. In particular,
for $t =T$, we have $U_T = u(0) e^{-\alpha (X_T + Q_T \psi)}$, and
since $U_T = u(0) e^{-\alpha(\pi^H - H)}$ we get that
\begin{equation*}
\pi^H - X_T - Q_T \psi = H \,.
\end{equation*}

\item[Claim 2:] As $H$ varies over all the $\cF_T$-measurable random
  variables, the left side of~\eqref{eq:30} sweeps for fixed
  $t=1,\ldots,T$ precisely across all the random variables $Z_t$ with
  $\condexp[\cF_{t-1}]{Z_t}=1$ which are of the form
  $Z_t = f_t(Y_1,\ldots,Y_t)$ for some function $f_t:\cbr{-1,+1}^t \to
  (0,\infty)$ such that
\[
0 < f_t(Y_1,\ldots,Y_{t-1},y_t) < 1/\condprob[\cF_{t-1}]{Y_t=y_t},
\quad y_t \in \cbr{-1,+1}\,.
\]

 Indeed, this is readily checked from 
\begin{align*}
\condexp[\cF_{t-1}]{Z_t}=&
f_t(Y_1,\dots,Y_{t-1},+1) \condprob[\cF_{t-1}]{Y_t=+1}\\&+f_t(Y_1,\dots,Y_{t-1},-1) \condprob[\cF_{t-1}]{Y_t=-1}.
\end{align*}

\item[Claim 3:] As $Q_t$ varies over the $\cF_{t-1}$-measurable random
  variables, the right side of~\eqref{eq:30} varies over the open interval
  bounded by $\underline{P}_t \set
  \frac{\condprob{\psi=\underline{\psi}_{t-1}}}{\condprob[\cF_{t-1}]{\psi=\underline{\psi}_{t-1}}}$
  and $\overline{P}_t \set
  \frac{\condprob{\psi=\overline{\psi}^{t-1}}}{\condprob[\cF_{t-1}]{\psi=\overline{\psi}^{t-1}}}$.

Indeed, this follows by the same argument as given in the proof for the second
claim in Proposition~\ref{prop:1}.
\end{description}

In light of Claims 1--3, completeness holds if and only if for any
$t=1,\dots,T$ we have $$\underline{P}_t \wedge \overline{P}_t = 0
\text{ and }
\underline{P}_t \vee \overline{P}_t =
1/\condprob[\cF_{t-1}]{Y_t=y_t}.$$
It follows that completeness yields $$\cbr{\underline{P}_t \leq
\overline{P}_t}\subseteq \cbr{\underline{P}_t =0}=
\cbr{\condprob{\psi=\underline{\psi}_{t-1}}=0}=\cbr{\underline{\psi}_{t-1}<\underline{\psi}_t}$$
and, similarly, $$\cbr{\underline{P}_t \geq \overline{P}_t}\subseteq \cbr{\overline{P}_t =0}=
\cbr{\condprob{\psi=\overline{\psi}^{t-1}}=0}=\cbr{\overline{\psi}^{t-1}>\overline{\psi}_t}.$$ In
particular, completeness implies~\eqref{eq:28}. Conversely, observe 
$$
\cbr{\overline{\psi}^{t-1}>\overline{\psi}^t} \cup \cbr{\underline{\psi}_{t-1}>\underline{\psi}_t} =\cbr{\underline{P}_t=0}\cup\cbr{\overline{P}_t=0}
$$
and thus~\eqref{eq:28} yields $\underline{P}_t \wedge
\overline{P}_t=0$.  To conclude we need to argue that $\underline{P}_t
\vee \overline{P}_t(\omega) = 1/\condprob[\cF_{t-1}]{Y_t=y_t}(\omega)$
for $\omega=(y_1,\dots,y_T) \in \Omega$. Without loss of generality,
it is sufficient to prove for $$\omega \in
\cbr{\underline{P}_t>0}=\cbr{\condprob{\psi=\underline{\psi}_{t-1}}>0}=\cbr{\underline{\psi}_t=\underline{\psi}_{t-1}}$$
we have in fact
$\underline{P}_t(\omega)=1/\condprob[\cF_{t-1}]{Y_t=y_t}(\omega)$. We
recall that
$$\condprob[\cF_{t-1}]{\psi=\underline{\psi}_{t-1}}=\condexp[\cF_{t-1}]{\condprob[\cF_{t}]{\psi=\underline{\psi}_{t-1}}}.$$
Thus $\underline{P}_t(\omega)=1/\condprob[\cF_{t-1}]{Y_t=y_t}(\omega)$ is equivalent to having 
$$\condprob[\cF_{t}]{\psi=\underline{\psi}_{t-1}}(\omega')=0$$ where
$\omega'=(y'_1,\dots,y'_T) \in \Omega$ is given by $y'_s \set y_s$ for $s
\not=t$ and $y'_t\set-y_t$. So assume for a contradiction
$\condprob[\cF_{t}]{\psi=\underline{\psi}_{t-1}}(\omega')>0$. Then we have 
$\underline{\psi}_{t}(\omega')=\underline{\psi}_{t-1}(\omega')$ in addition to
$\underline{\psi}_{t}(\omega)=\underline{\psi}_{t-1}(\omega)$. By~\eqref{eq:28},
this implies
$\overline{\psi}_{t}(\omega)>\overline{\psi}_{t-1}(\omega)$ and
$\overline{\psi}_{t}(\omega')>\overline{\psi}_{t-1}(\omega')$ in
contradiction to
$\overline{\psi}_{t-1}(\omega)=\overline{\psi}_{t-1}(\omega')=\overline{\psi}_{t}(\omega)
\wedge \overline{\psi}_{t}(\omega')$.
\end{proof}

\begin{Remark} 
  For the model with binomial lattice structure considered in this
  section, one can show by similar arguments that there exists a
  superreplicating strategy if
\begin{equation*}
\emptyset \neq \cbr{\underline{\psi}_{t-1}<\underline{\psi}_{t}} \cup
  \cbr{\overline{\psi}^{t-1}>\overline{\psi}^{t}}, \; t=1,\ldots,T\,.
\end{equation*}
Moreover, in Section ~\ref{sec:counterexample} we provide an example,
where superreplicating cost is not attained in the case
\begin{equation*}
\emptyset = \cbr{\underline{\psi}_{0}<\underline{\psi}_{1}} \cup
  \cbr{\overline{\psi}^{0}>\overline{\psi}^{1}}.
\end{equation*}
\end{Remark}

\section{Examples of processes with decreasing exponential tails}
\label{sec:examples}

To illustrate that our assumption of decreasing exponential tails is
satisfied by typical models considered for stock prices, let us study
a few examples.

\subsection{L\'{e}vy processes}

\begin{Lemma} \label{lem:5}
Let $(X_t)_{t \geq 0}$ be a non-deterministic one-dimensional L\'{e}vy process with respect
to the filtration $(\cF_t)_{t \geq 0}$. Suppose that $X_T$ has all exponential moments, i.e. for all $q \in
\mathbb{R}$, $\E \left[ e^{qX_T} \right] < \infty$. Then $$\psi \set X_T$$
exhibits decreasing exponential tails in the sense of condition~\eqref{eq:12}.
\end{Lemma}
\begin{proof}
Let us denote the L\'{e}vy triplet of $X$ by $(b, c, \mu)$, where $b \in \mathbb{R}$, $c \geq 0$ and $\mu$ is a L\'{e}vy measure on $\mathbb{R}$ satisfying
\begin{equation} \label{eq:31}
\mu(\cbr{ 0 }) = 0, \quad \int_{\mathbb{R}} \left( 1 \wedge |x|^2
\right) \, \mu(dx) < \infty.
\end{equation}
We exclude the deterministic case where both $c=0$ and $\mu=0$ and we
require that $X_T$ has all exponential moments, i.e. for all $q \in
\mathbb{R}$, $\E \left[ e^{qX_T} \right] < \infty$. This is equivalent
to stating 
\begin{equation}\label{eq:32}
\int_{\cbr{|x| \geq 1 }} e^{qx} \, \mu(dx) < \infty, \; q \in \RR.
\end{equation}

Since $X$ has homogeneous and independent increments we will have
$\nu_{t+h} \prec \nu_t$ for $h>0$ where $\nu_s \set \condprob[\cF_s]{X_T \in
  \cdot}$, $s \in [0,T]$, if
\begin{equation}\label{eq:33}
\frac{\E \left[ e^{q\psi} | \cF_{t+h}  \right]}{\E \left[ e^{q\psi} | \cF_t
  \right]}=\exp(q (X_{t+h}-X_t)-h f(q)) \to 0 \text{ as } |q| \uparrow \infty
\end{equation}
where due to the L\'{e}vy-Khintchine formula 
\begin{equation*}
f(q) = bq + \frac{1}{2} cq^2 + \int_{\mathbb{R}} \cbr{ e^{qx} - 1- qx \mathbf{1}_{\cbr{|x| < 1}} } \,\mu(dx).
\end{equation*}
By considering first $X$ and then $-X$ we can confine ourselves to
considering the limit as $q \uparrow \infty$.
\begin{description}
\item[Case  `$\mu((0,\infty)>0$':]
 We shall show that $\int_{\mathbb{R}} \cbr{ e^{qx} - 1- qx
   \mathbf{1}_{\cbr{|x| < 1}} } \,\mu(dx)$ converges to $\infty$
 exponentially fast as $q \uparrow \infty$ which entails the same
 convergence for $f(q)$ and thus proves~\eqref{eq:33}.

 Indeed, if $\epsilon \in (0,1)$ is chosen such that
 $\mu((\epsilon,\infty))>0$ then
 \begin{equation}\label{eq:34}
 \int_{(\epsilon,\infty)} \cbr{ e^{qx} - 1- qx \mathbf{1}_{\cbr{|x| <
     1}} } \,\mu(dx) \geq (e^{q \epsilon}-(1+q)) \, \mu((\epsilon,\infty))
\end{equation}
converges to $\infty$ exponentially fast as $q \uparrow \infty$. At
the same time we can use the Taylor expansion
\[
e^{q x}= 1 + qx + \frac{1}{2} z^2 \text{ for some } z=z(qx) \text{
  between $0$ and $qx$}
\]
to see first that
 \begin{equation}\label{eq:35} 
\int_{(0,\epsilon)} \cbr{ e^{qx} - 1- qx \mathbf{1}_{\cbr{|x| <
     1}} } \,\mu(dx) \geq 0
\end{equation}
and second that
 \begin{equation}\label{eq:36} 
\int_{(-\infty,0)} \cbr{ e^{qx} - 1- qx \mathbf{1}_{\cbr{|x| <
     1}} } \,\mu(dx) \geq \int_{(-\infty,0)} qx \mathbf{1}_{\cbr{|x| \geq
     1}} \,\mu(dx)\,.
\end{equation}
Since the exponential moment condition~\eqref{eq:32} implies $\int |x|
\mathbf{1}_{\cbr{|x| \geq 1}} \,\mu(dx)<\infty$, the latter expression
may diverge to $-\infty$ at most linearly fast in $q$. The summation
of~\eqref{eq:34}, \eqref{eq:35} and~\eqref{eq:36} thus proves the
claimed exponential divergence of $\int_{\mathbb{R}} \cbr{ e^{qx} -
  1- qx \mathbf{1}_{\cbr{|x| < 1}} } \,\mu(dx)$ to $\infty$.

\item[Case `$\mu((0,\infty)) = 0$':] In this case we infer
  from~\eqref{eq:36} that the integral in the definition of $f(q)$
  may converge to $-\infty$ for $q \uparrow \infty$ at most at the
  deterministic linear rate $r \set \int_{(-\infty,0)} |x|
  \mathbf{1}_{\cbr{|x| \geq 1}} \,\mu(dx)<\infty$. Hence, $f(q)$ will
  diverge to $+\infty$ quadratically in $q$ if $c>0$ in which
  case~\eqref{eq:34} holds on $\Omega$. If $c=0$, $f(q)$ may
  converge to $-\infty$ at most with linear speed $|b|+r$. Hence,
  recalling that the argument for the previous case shows that  the limit $q \downarrow -\infty$
  in~\eqref{eq:33} holds on $\Omega$, it suffices to observe
  that~\eqref{eq:33} will hold for $q \uparrow \infty$ on $\cbr{X_{t+h}-X_t-h(|b|+r)<0}$, a set
  with positive probability if $c=0$, $\mu(0,\infty)=0$ unless $X$ is
  deterministic which we have ruled out from the start. 
\end{description}

\end{proof}

\subsection{A model of Barndorff-Nielsen Shephard-type}

Clearly, Laplace transforms can be computed in many financial models
and so our condition of decreasing exponential tails can be checked in
more complex models as well. By way of illustration let us consider a
stochastic volatility model in the style of \citet{bns2001}.

\begin{Lemma} \label{lem:6}
Let $Z = \left( Z_t \right)_{t \geq0}$ be a L\'evy subordinator 
on a filtered probability space $(\Omega,\cF,(\cF_t),\P)$ with L\'evy
measure $\mu$ such that
\begin{equation}\label{eq:37}
\kappa(\theta) \set \log \left( \E \left[ e^{\theta Z_1} \right]
\right) = \int_0^{\infty} \left( e^{\theta y} -1 \right)
\mu(dy)<\infty, \; \theta \in \RR
\end{equation}
and $W = \left( W_t \right)_{t \geq 0}$ be an independent Brownian motion of $Z$.  
Assume the payoff of the marketed claim is given by $$\psi = X_T\,$$
where $(X_t)_{t \in [0,T]}$ follows the Barndorff-Nielsen Shephard
dynamics
\begin{align*}
dX_t & = \left( m + \beta \sigma_t^2 \right) dt + \sigma_t dW_t + \rho dZ_{\lambda t}, \\
d\sigma_t^2 & = - \lambda \sigma_t^2 dt + dZ_{\lambda t}, \quad \sigma_0^2 > 0.
\end{align*}
for some constants $m, \beta\in \mathbb{R}$, $\lambda > 0$ and $\rho < 0$.
Then $\psi$ exhibits decreasing exponential tails in the sense of condition~\eqref{eq:12}.
\end{Lemma}

\begin{proof}

The Laplace transform $\E \left[ e^{q X_T} | \cF_t \right]$ is computed in \citet{nv2003}:
\begin{align} \label{eq:38}
\condexp{e^{q X_T}} &=\\ \nonumber\quad \exp &\left( q \left( X_t + m(T-t) \right) + \left( q^2 + 2 \beta q \right) \frac{\epsilon(t,T)}{2}\sigma_t^2 + \int_t^T \lambda \kappa\left( f(s,q) \right) ds \right),
\end{align}
where
\begin{align*}
\epsilon(s,T) & \set \frac{1}{\lambda} \left( 1 - e^{-\lambda(T-s)} \right), \\
f(s,q) & \set \rho q + \frac{1}{2} \left( q^2 + 2\beta q \right) \epsilon(s,T).
\end{align*}

Moreover, according to Theorem~2 in \citet{nv2003}, under condition~\eqref{eq:37} $\E \left[ e^{q X_T} | \cF_t \right] < \infty$ for every $q \in \mathbb{R}$, i.e. $\psi=X_T$ has all exponential moments. 

We claim that for any $h>0$ we have $\nu_{t+h} \prec \nu_t$ where as
before $\nu_s \set \condprob[\cF_s]{\psi \in \cdot}$. 
From \eqref{eq:38} it suffices to show 
\begin{align} \label{eq:39}
\lim_{|q| \rightarrow \infty} \bigg\{q \left( X_{t+h} - X_t - mh \right)  & + \left( q^2 + 2\beta q\right)  \left( \frac{\epsilon(t+h,T)}{2}\sigma_{t+h}^2 - \frac{\epsilon(t,T)}{2}\sigma_{t}^2 \right) \\ \notag
& - \int_t^{t+h} \lambda \kappa(f(s,q)) ds  \bigg\} = - \infty.
\end{align}
To show this, we note that by Taylor expansion
\begin{align*}
\kappa(f(s,q)) &= \int_0^{\infty} \left( e^{f(s,q) y} - 1 \right)
\mu(dy) \\
&\geq \int_0^{\infty} \left(  f(s,q) y + \frac{1}{2} f(s,q)^2 y^2 + \frac{1}{6} f(s,q)^3 y^3  \right) \mu(dy)\,.
\end{align*}
For fixed $s \in [t, t+h]$ and $y > 0$ we obtain by direct calculation
\begin{align*}
F(s,q,y) &\set f(s,q) y + \frac{1}{2} f(s,q)^2 y^2 + \frac{1}{6}
f(s,q)^3 y^3 \\
&= \frac{1}{48} q^6 \epsilon(s,T)^3 y^3 + P_5(s, q, y),
\end{align*}
where $P_5$ is a polynomial of order $5$ in $q$. The coefficients of
this polynomial $P_5$ are functions of $s$ and $y$, where the
dependence of $P_5$ on $s$ is continuous. Moreover, since
$\kappa(\theta) < \infty$ for all $\theta$, we have that
\begin{equation*}
\int_0^{\infty} y^n \mu(dy) < \infty, \quad n = 1,2,\ldots.
\end{equation*}
Hence,
\begin{align*}
  \lambda \int_t^{t+h}  \kappa(f(s,q)) ds & \geq  \lambda \int_t^{t+h} \int_0^{\infty} F(s,q,y) \mu(dy) ds 
\end{align*}
and the latter expression is a polynomial of order $6$ in $q$ with
positive leading coefficient. Thus, the integral term dominates all
other terms in~\eqref{eq:39} and we obtain the claimed convergence.
\end{proof}

\bibliographystyle{plainnat}
\bibliography{mybib}

\begin{thebibliography}{24}
\providecommand{\natexlab}[1]{#1}
\providecommand{\url}[1]{\texttt{#1}}
\expandafter\ifx\csname urlstyle\endcsname\relax
  \providecommand{\doi}[1]{doi: #1}\else
  \providecommand{\doi}{doi: \begingroup \urlstyle{rm}\Url}\fi

\bibitem[Bank and Kramkov(2011{\natexlab{a}})]{bk2011a}
Peter Bank and Dmitry Kramkov.
\newblock A model for a large investor trading at market indifference prices.
  {I}: single-period case.
\newblock \emph{preprint}, 2011{\natexlab{a}}.
\newblock URL \url{http://arxiv.org/abs/1110.3224v2}.

\bibitem[Bank and Kramkov(2011{\natexlab{b}})]{bk2011b}
Peter Bank and Dmitry Kramkov.
\newblock A model for a large investor trading at market indifference prices.
  {II}: continuous-time case.
\newblock \emph{preprint}, 2011{\natexlab{b}}.
\newblock URL \url{http://arxiv.org/abs/1110.3229v2}.

\bibitem[Bank and Kramkov(2013)]{bk2013}
Peter Bank and Dmitry Kramkov.
\newblock The stochastic field of aggregate utilities and its saddle conjugate.
\newblock \emph{preprint}, 2013.

\bibitem[Barndorff-Nielsen and Shephard(2001)]{bns2001}
Ole~E Barndorff-Nielsen and Neil Shephard.
\newblock Non-{G}aussian {O}rnstein--{U}hlenbeck-based models and some of their
  uses in financial economics.
\newblock \emph{Journal of the Royal Statistical Society: Series B (Statistical
  Methodology)}, 63\penalty0 (2):\penalty0 167--241, 2001.

\bibitem[Broadie et~al.(1998)Broadie, Cvitanic, and Soner]{bcs1998}
Mark Broadie, Jaksa Cvitanic, and H~Mete Soner.
\newblock Optimal replication of contingent claims under portfolio constraints.
\newblock \emph{Review of Financial Studies}, 11\penalty0 (1):\penalty0 59--79,
  1998.

\bibitem[Campi and Schachermayer(2006)]{cs2006}
Luciano Campi and Walter Schachermayer.
\newblock A super-replication theorem in {K}abanovÕ{s} model of transaction
  costs.
\newblock \emph{Finance and Stochastics}, 10\penalty0 (4):\penalty0 579--596,
  2006.

\bibitem[\c{C}etin et~al.(2004)\c{C}etin, Jarrow, and Protter]{cjp2004}
Umut \c{C}etin, Robert~A Jarrow, and Philip Protter.
\newblock Liquidity risk and arbitrage pricing theory.
\newblock \emph{Finance and Stochastics}, 8\penalty0 (3):\penalty0 311--341,
  2004.

\bibitem[{\c{C}}etin et~al.(2010){\c{C}}etin, Soner, and Touzi]{cst2010}
Umut {\c{C}}etin, H~Mete Soner, and Nizar Touzi.
\newblock Option hedging for small investors under liquidity costs.
\newblock \emph{Finance and {S}tochastics}, 14\penalty0 (3):\penalty0 317--341,
  2010.

\bibitem[Cvitani{\'c} and Karatzas(1993)]{ck1993}
Jak{\v{s}}a Cvitani{\'c} and Ioannis Karatzas.
\newblock Hedging contingent claims with constrained portfolios.
\newblock \emph{The {A}nnals of {A}pplied {P}robability}, pages 652--681, 1993.

\bibitem[Cvitani{\'c} et~al.(1999)Cvitani{\'c}, Pham, and Touzi]{cpt1999}
Jak{\v{s}}a Cvitani{\'c}, Huyen Pham, and Nizar Touzi.
\newblock A closed-form solution to the problem of super-replication under
  transaction costs.
\newblock \emph{Finance and {S}tochastics}, 3\penalty0 (1):\penalty0 35--54,
  1999.

\bibitem[Dolinsky and Soner(2013)]{ds2013}
Yan Dolinsky and Halil~Mete Soner.
\newblock Duality and convergence for binomial markets with friction.
\newblock \emph{Finance and Stochastics}, pages 1--29, 2013.

\bibitem[El~Karoui and Quenez(1995)]{eq1995}
Nicole El~Karoui and Marie-Claire Quenez.
\newblock Dynamic programming and pricing of contingent claims in an incomplete
  market.
\newblock \emph{SIAM journal on {C}ontrol and {O}ptimization}, 33\penalty0
  (1):\penalty0 29--66, 1995.

\bibitem[F{\"o}llmer and Kramkov(1997)]{fk1997}
Hans F{\"o}llmer and Dmitry Kramkov.
\newblock Optional decompositions under constraints.
\newblock \emph{Probability {T}heory and {R}elated {F}ields}, 109\penalty0
  (1):\penalty0 1--25, 1997.

\bibitem[F{\"o}llmer and Schied(2011)]{fs2011}
Hans F{\"o}llmer and Alexander Schied.
\newblock \emph{Stochastic {F}inance}.
\newblock Walter de Gruyter \& Co., extended edition, 2011.
\newblock ISBN 978-3-11-021804-6.
\newblock An introduction in discrete time.

\bibitem[G{\"o}kay and Soner(2012)]{gs2012}
Selim G{\"o}kay and Halil~Mete Soner.
\newblock Liquidity in a binomial market.
\newblock \emph{Mathematical Finance}, 22\penalty0 (2):\penalty0 250--276,
  2012.

\bibitem[Guasoni et~al.(2008)Guasoni, R{\'a}sonyi, and Schachermayer]{grs2008}
Paolo Guasoni, Mikl{\'o}s R{\'a}sonyi, and Walter Schachermayer.
\newblock Consistent price systems and face-lifting pricing under transaction
  costs.
\newblock \emph{The {A}nnals of {A}pplied {P}robability}, 18\penalty0
  (2):\penalty0 491--520, 2008.

\bibitem[Jouini and Kallal(1995)]{jk1995}
Ely{\'e}gs Jouini and H{\'e}di Kallal.
\newblock Arbitrage in securities markets with short-sales constraints.
\newblock \emph{Mathematical {F}inance}, 5\penalty0 (3):\penalty0 197--232,
  1995.

\bibitem[Kabanov et~al.(2002)Kabanov, R{\'a}sonyi, and Stricker]{krs2002}
Yuri Kabanov, Mikl{\'o}s R{\'a}sonyi, and Christophe Stricker.
\newblock No-arbitrage criteria for financial markets with efficient friction.
\newblock \emph{Finance and Stochastics}, 6\penalty0 (3):\penalty0 371--382,
  2002.

\bibitem[Kabanov and Safarian(2009)]{ks2009}
Yuri~M Kabanov and Mher~M Safarian.
\newblock \emph{Markets with transaction costs}.
\newblock Springer, 2009.

\bibitem[Kabanov and Stricker(2002)]{ks2002}
Yuri~M Kabanov and Christophe Stricker.
\newblock Hedging of contingent claims under transaction costs.
\newblock In \emph{Advances in Finance and Stochastics}, pages 125--136.
  Springer, 2002.

\bibitem[Kramkov(1996)]{k1996}
Dmitry~O Kramkov.
\newblock Optional decomposition of supermartingales and hedging contingent
  claims in incomplete security markets.
\newblock \emph{Probability {T}heory and {R}elated {F}ields}, 105\penalty0
  (4):\penalty0 459--479, 1996.

\bibitem[Nicolato and Venardos(2003)]{nv2003}
Elisa Nicolato and Emmanouil Venardos.
\newblock Option pricing in stochastic volatility models of the
  {O}rnstein-{U}hlenbeck type.
\newblock \emph{Mathematical Finance}, 13\penalty0 (4):\penalty0 445--466,
  2003.

\bibitem[Possamai et~al.(2012)Possamai, Soner, and Touzi]{pst2012}
Dylan Possamai, H~Mete Soner, and Nizar Touzi.
\newblock Large liquidity expansion of super-hedging costs.
\newblock \emph{Asymptotic {A}nalysis}, 79\penalty0 (1):\penalty0 45--64, 2012.

\bibitem[Soner et~al.(1995)Soner, Shreve, and Cvitani{\'c}]{ssc1995}
Halil~M Soner, Steven~E Shreve, and J~Cvitani{\'c}.
\newblock There is no nontrivial hedging portfolio for option pricing with
  transaction costs.
\newblock \emph{The {A}nnals of {A}pplied {P}robability}, pages 327--355, 1995.

\end{thebibliography}

\end{document}